\documentclass[11pt,a4paper,reqno]{amsart}
\usepackage[hmargin={2.7cm,2.7cm},vmargin={2.5cm,2.5cm},centering]{geometry}
\usepackage{amssymb,amsmath,amsthm,amsfonts,amscd}
\usepackage{graphicx}
\usepackage[dvipsnames]{xcolor}
\usepackage[unicode,psdextra]{hyperref}
\usepackage[utf8]{inputenc}
\usepackage{enumerate}
\hypersetup{pdfborder={0 0 0.06},linkbordercolor=BrickRed,citebordercolor=Fuchsia,urlbordercolor=CornflowerBlue}
\usepackage{tikz-cd}
\usepackage[english]{babel}
\usepackage{dsfont}
\usepackage{url}
\usepackage[longnamesfirst,numbers,sort&compress]{natbib}
\usepackage{setspace}
\usepackage[font={scriptsize,bf}]{caption}
\usepackage{changepage}
\usepackage{booktabs}
\usepackage[sort&compress,capitalise,nameinlink]{cleveref}
\crefname{section}{\textsection}{\textsection}
\crefname{subsection}{\textsection}{\textsection}
\crefname{subsubsection}{\textsection}{\textsection}
\crefname{paragraph}{\textparagraph}{\textparagraph}
\crefname{thm}{Theorem}{Theorems}
\crefname{minprob}{Minimization Problem}{Minimization Problems}
\usepackage{mathrsfs}
\usepackage{bm}

\DeclareMathOperator{\tr}{tr}

\DeclareMathOperator*{\wstarlim}{w*-lim}
\renewcommand{\Im}{\mathrm{Im}}
\renewcommand{\Re}{\mathrm{Re}}
          \newtheorem{thm}{Theorem}[section]
          \newtheorem{proposition}[thm]{Proposition}
          \newtheorem{lemma}[thm]{Lemma}
          \newtheorem{corollary}[thm]{Corollary}

          \theoremstyle{definition}

          \newtheorem{remark}[thm]{Remark}

\usepackage{csquotes}
\usepackage[authormarkup=none,
deletedmarkup=xout
]{changes}

\newcommand{\diff}{\mathrm{d}}
\newcommand{\eps}{\varepsilon}

\newcommand{\beq}{\begin{equation}}
\newcommand{\eeq}{\end{equation}}

\newcommand{\bdm}{\begin{displaymath}}
\newcommand{\edm}{\end{displaymath}}
\newcommand{\bdn}{\begin{eqnarray}}
\newcommand{\edn}{\end{eqnarray}}
\newcommand{\bay}{\begin{array}{c}}
\newcommand{\eay}{\end{array}}
\newcommand{\ben}{\begin{enumerate}}
\newcommand{\een}{\end{enumerate}}
\newcommand{\beqn}{\begin{eqnarray}}
\newcommand{\eeqn}{\end{eqnarray}}

\newcommand{\lf}{\left}
\newcommand{\ri}{\right}

\renewcommand{\leq}{\leqslant}
\renewcommand{\geq}{\geqslant}

\newcommand{\verde}[1]{\textcolor{Green}{#1}}

\numberwithin{equation}{section}

\begin{document}
\title{Quasi-Classical Spin Boson Models}

\author[M.\ Correggi]{Michele Correggi}

\address{Dipartimento di Matematica, Politecnico di Milano, Piazza Leonardo
  da Vinci 32, 20133, Milano, Italy.}

\email{michele.correggi@gmail.com}

\urladdr{\url{https://sites.google.com/view/michele-correggi}}

\author[M.\ Falconi]{Marco Falconi}

\address{Dipartimento di Matematica, Politecnico di Milano, Piazza Leonardo
  da Vinci 32, 20133, Milano, Italy.}

\email{marco.falconi@polimi.it}

\urladdr{\url{https://www.mfmat.org/}}

\author[M.\ Merkli]{Marco Merkli}

\address{Department of Mathematics and Statistics, Memorial University of
  Newfoundland, St. John's, NL A1C 5S7, Canada.}

\email{merkli@mun.ca}

\urladdr{\url{https://www.math.mun.ca/~merkli/}}




\begin{abstract}
  In this short note we study Spin-Boson Models from the Quasi-Classical
  standpoint. In the Quasi-Classical limit, the field becomes macroscopic
  while the particles it interacts with, they remain quantum. As a result,
  the field becomes a classical environment that drives the particle system
  with an explicit effective dynamics.
\end{abstract}

\maketitle

\onehalfspacing{}


\section{Introduction and Main Result}
\label{sec:intr-main-results}

The spin-boson model describes the interaction between a
bosonic scalar field, playing the role of environment or
  reservoir, and a `small' quantum system, whose spin degrees of
  freedom are the only relevant ones. It has been widely studied in the
mathematical physics and physics literature, from various standpoints. The
spin-boson model is one of {\em the} paradigmatic examples of an open quantum
system. It is used to investigate general open system phenomena such as
decoherence, entanglement, thermalization, to test the validity of markovian
approximations and to analyze non-markovian behavior. We cannot attempt to
give an exhaustive list of references of the model. We point the
mathematically interested reader to the following inconclusive list of works,
\citep[][]{amour2016alnarxiv,amour2017arxiv,amour2017jmp,arai1990jmp,arai1991jmp,arai1997jfa,hasler2011ahp},
as well as to references therein contained.

On the more physical side, the spin-boson model is used to describe
atom-radiation interaction in quantum optics, qubit-noise coupling in quantum
information and computation, environment induced transport phenomena and
chemical processes in quantum chemistry. Some of these aspects can be found
in the references
\citep{leggett1987rmp,palma1996prsla,koenenberg2015cmp,xu1994cp,konenberg2016jmp,mohseni2014cup,merkli2020ap,merkli2022qI,merkli2022qII,joye2020ahp,merkli2016jmc,merkli2013jmc,merkli2011qic,merkli2007prl}.

\medskip

For the purpose of this paper, in which we focus on mathematical aspects, we
assume that the reader is familiar with the basic mathematical tools of free
quantum fields, namely Fock spaces, second quantization,
creation/annihilation operators, etc.; if not, they may refer, \emph{e.g.},
to \citep{cook1951pnasusa, derezinski2013cmmp}. Let us denote by
$\mathscr{H}$ the Hilbert space of the spin system, and by $\mathfrak{h}$ the
Hilbert space of a single bosonic excitation.  We denote by $\mathcal{G}_{\varepsilon}$
the second quantization functor\footnote{We use a somewhat unorthodox
  notation for the second quantization functor. We denote by $\mathcal{G}_{\varepsilon}^{\rm
    s}(\mathfrak{h})=\bigoplus_{n\in \mathbb{N}} \mathfrak{h}_n$ the symmetric Fock space over $\mathfrak{h}$ in which the
  canonical creation and annihilation operators have $\varepsilon$-dependent
  commutation relations:
  \begin{equation*}
    [a_{\varepsilon}(f),a_{\varepsilon}^{*}(g)]= \varepsilon\langle f  , g \rangle_{\mathfrak{h}}\;.
  \end{equation*}
  The second quantization of an operator $A$ on $\mathfrak{h}$ is written thus as
  \begin{equation*}
    \mathrm{d}\mathcal{G}_{\varepsilon}(A)= \sum_{i,j=0}^{\infty}A_{ij}a^{*}_{\varepsilon,i}a_{\varepsilon,j}\;,
  \end{equation*}
  with $A_{ij}=\langle e_i , A e_j \rangle_{\mathfrak{h}}$, and $a^{\sharp}_{\varepsilon,k}=a_{\varepsilon}^{\sharp}(e_k)$, with
  $\{e_k\}_{k\in \mathbb{N}}$ an O.N.B.\ of $\mathfrak{h}$. The quasi-classical parameter $\varepsilon$
  clearly plays the role of a semiclassical parameter for the (Segal) field
  $\varphi_{\varepsilon}(f)= a^{*}_{\varepsilon}(f)+ a_{\varepsilon}(f)$: as $\varepsilon\to 0$, the field becomes a
  classical commutative observable \citep[see][for a gentler and more
  detailed introduction to the quasi-classical scaling]{correggi2019arxiv},
  and
  \citep{carlone2021sima,correggi2017ahp,correggi2017arxiv,correggi2020arxiv}
  for other recent papers concerning the quasi-classical regime.},
  where $ 0 < \eps \ll 1 $ is a scale parameter. The spin-boson Hamiltonian
has the general form
\begin{equation*}
  H_{\varepsilon}= \mathfrak{S} \otimes 1 + \nu(\varepsilon)\, 1\otimes \mathrm{d}\mathcal{G}_{\varepsilon}(\omega) + \mathfrak{s}\otimes \varphi_{\varepsilon}(g)\;,
\end{equation*}
as an operator on $\mathscr{H}\otimes \mathcal{G}_{\varepsilon}^{\rm
  s}(\mathfrak{h})$. Here, $\mathfrak{S},\mathfrak{s}\in
\mathscr{B}(\mathscr{H})$ are self-adjoint, $\nu(\varepsilon)$ is
either $\nu(\varepsilon)=1$ or $\nu(\varepsilon)=\frac{1}{\varepsilon}$,
$\omega$ is a positive -- with possibly unbounded inverse -- operator on
$\mathfrak{h}$, and $g\in \mathfrak{h}$. The bipartition of the total Hilbert
space $\mathscr H\otimes\mathfrak \mathcal G_\varepsilon^{\rm
  s}(\mathfrak h)$ reflects the separation of the total physical system into
two subsystems. Commonly, especially in the physics literature, the Hilbert
space $\mathscr H$ is finite-dimensional. For instance,
$\mathscr H$ has dimension $2^N$ in the case of $N$ spins $1/2$, or
$N$ qubits. One of the most studied cases is $N=2$, hence the name
``spin-boson'' model. A further characteristic of the spin-boson model is
that the interaction operator is of the simple product form $\mathfrak
s\otimes\varphi_\varepsilon(g)$, or a finite sum of such terms. This
simplifies the (rigorous) analysis of the model. Nevertheless, other models
in which the interaction term is more complicated, are also of interest. For
instance, in the Nelson, the Pauli-Fierz or the polaron model, the
interaction operator is of the form $\int_{\mathbb R^3}\mathfrak
s(k)\otimes\varphi_\varepsilon(k)\mathrm d^3k$. While these models are also
treatable with the methods explained here (see \citep{correggi2019arxiv}), we
focus in the present manuscript, for ease of presentation, on the simple form
of the interaction as in $H_\varepsilon$ above.

We shall consider a more general setup though. The guiding principle is that
we want to describe two {\em qualitatively unequal} interacting parts. The
`spin' part which is `small' and the boson part (or field, reservoir,
environment) which is `large'. A quantification of what small versus large
means can be implemented in different ways, depending on the physical reality
being modeled. For instance, finite dimensional ($\mathscr H$) versus
infinite dimensional ($\mathcal G^{\rm s}_{\varepsilon}(\mathfrak{h})$)
Hilbert spaces, or Hamiltonians with discrete spectrum ($\mathfrak S$) versus
Hamiltonians with continuous spectrum
($\mathrm{d}\mathcal{G}_{\varepsilon}(\omega)$). In the quasi-classical setup
we are discussing here, the field is large in the sense that it is in a state
which contains many more particles (or excitations) than the spin system
does. This is formalized by saying that the (average) number of particles of
the spin system is fixed, while that number in the field state is $\propto
1/\varepsilon \gg 1$ -- constituting the quasi-classical limit.

We will soon explain how the choice of $\nu(\varepsilon)$ affects the
quasi-classical limit $\varepsilon\to 0$. There are further possible
generalizations of the model, namely by taking $\mathfrak{s}$ and $g$ to be
vector-valued, or by taking $\mathfrak{S}$ to be only bounded from below, or
by taking $g\notin \mathfrak{h}$. Self-adjointness for the latter case has
been recently studied in \citep{lonigro2022jmp}. For our purposes, these
generalizations do not present serious obstacles as long as $H_{\varepsilon}$
can be defined as a self-adjoint operator, however for the sake of clarity we
keep the setting as described above.

\begin{proposition}[Self-adjointness of $ H_{\eps} $]
  \label{prop:1}
  \mbox{}   \\
  For all $g\in \mathfrak{h}$, $H_{\varepsilon}$ is essentially self-adjoint
  on\footnote{We denote by
    $\mathcal{C}_0^{\infty}\bigl(\mathrm{d}\mathcal{G}_{\varepsilon}(1)\bigr)$
    the Fock space vectors with a finite number of particles (\emph{i.e.},
    for which the $k$-particle components are all zero for $k>\underline{k}$,
    for some $\underline{k}\in \mathbb{N}$).}
  $D\bigl(\mathrm{d}\mathcal{G}_{\varepsilon}(\omega)\bigr)\cap
  \mathcal{C}_0^{\infty}\bigl(\mathrm{d}\mathcal{G}_{\varepsilon}(1)\bigr)$. In
  addition, if both $g\in \mathfrak{h}$ and $\omega^{-1/2}g\in \mathfrak{h}$,
  then $H_{\varepsilon}$ is self-adjoint on
  $D\bigl(\mathrm{d}\mathcal{G}_{\varepsilon}(\omega)\bigr)$ and bounded from
  below.
\end{proposition}
\begin{proof}
  The essential self-adjointness is proved in \citep{falconi2015mpag}, for a
  general class of operators describing the interaction between matter and
  radiation; self-adjointness and boundedness from below with the additional
  assumption $\omega^{-1/2}g\in \mathfrak{h}$ is an easy consequence of the
  Kato-Rellich theorem on relatively bounded perturbations of self-adjoint
  operators.
\end{proof}

\begin{remark}[Form factors]
  \label{rem:1}
  \mbox{}   \\
  There are form factors $g\in \mathfrak{h}$ with $\omega^{-1/2}g\notin
  \mathfrak{h}$ such that $H_{\varepsilon}$ is unbounded from below (even
  though it is still self-adjoint). This is analogous to what happens for the
  van Hove model, and it is caused by some infrared singularity: in physical
  models, $\omega^{-1/2}$ is unbounded (and thus it could happen that
  $\omega^{-1/2}g\notin \mathfrak{h}$) only if the field is massless
  \citep[see][for further details]{derezinski2003ahp}. For our purposes,
  uniqueness of the quantum dynamics (\emph{i.e.} essential self-adjointness
  of $H_{\varepsilon}$) is enough.
\end{remark}

\subsection{Main Result}
\label{sec:main-results}

Our goal is to characterize explicitly the dynamics of quantum states, in the
limit $\varepsilon\to 0$. In order to do that, let us define quantum states
as density matrices
\begin{equation*}
  \Gamma_{\varepsilon}\in \mathfrak{L}^1_{+,1}\bigl(\mathscr{H}\otimes \mathcal{G}_{\varepsilon}^{\rm s}(\mathfrak{h})\bigr)\;,
\end{equation*}
where $\mathfrak{L}^1$ is the trace ideal, and $\mathfrak{L}^1_{+,1}$ stands
for elements in the positive cone, with trace one. A time-evolved state is
then given by
\begin{equation*}
  \Gamma_{\varepsilon}(t)= e^{-it H_{\varepsilon}}\Gamma_{\varepsilon}e^{it H_{\varepsilon}}\;.
\end{equation*}
To be more precise, the question we will answer in this note is the
following:
\begin{quote}
  \emph{Knowing the behavior of the initial state $\Gamma_{\varepsilon}$ as
    $\varepsilon\to 0$, what is the behavior of $\Gamma_{\varepsilon}(t)$ as
    $\varepsilon\to 0$, for any time $t\in \mathbb{R}$?}
\end{quote}
To answer the question, we shall first clarify what the general behavior is
of a quantum state $\Gamma_{\varepsilon}$, as $\varepsilon\to 0$. The
intuition is that as the boson degrees of freedom become classical, the state
-- restricted to the boson subsystem -- becomes classical as well (in the
statistical mechanics sense,\emph{ i.e.} a probability measure); on the other
hand, the spin subsystem retains its quantum nature, and thus its description
shall still be given by a density matrix.

This picture is satisfactorily described mathematically in terms of a
so-called \emph{state-valued measure}, introduced in
\citep{correggi2019arxiv,falconi2017arxiv}. A state-valued measure is a
couple $\mathfrak{m}=(\mu,\gamma)$ consisting of a (Borel Radon) measure
$\mu$ on the classical configuration space $\mathfrak{h}$ for the Boson
subsystem, and a $\mu$-almost-everywhere defined function $\mathfrak{h}\ni z
\mapsto \gamma(z)\in \mathfrak{L}^1_{+,1}(\mathscr{H})$ with values in the
density matrices of the Spin subsystem. The function $\gamma(z)$ acts as a
vector-valued Radon-Nikod{\'y}m derivative (it is in fact one), and thus the
measure element $\mathrm{d}\mathfrak{m}(z)$ can be written as
\begin{equation*}
  \mathrm{d}\mathfrak{m}(z)=\gamma(z)\mathrm{d}\mu(z)\;.
\end{equation*}
Integrating a scalar measurable bounded function $F$ with respect to
$\mathfrak{m}$ gives an element in $\mathfrak{L}^1(\mathscr{H})$, that we
denote by
\begin{equation*}
  \int_{\mathfrak{h}}^{}F(z)  \mathrm{d}\mathfrak{m}(z)= \int_{\mathfrak{h}}^{}F(z)\gamma(z)  \mathrm{d}\mu(z)\;.
\end{equation*}
It is also possible to integrate suitable functions $\mathfrak{F}$ with
values in the bounded operators on $\mathscr{H}$, however in this case the
relative order between the function and the measure matters: in general,
\begin{gather*}
  \int_{\mathfrak{h}}^{}\mathfrak{F}(z)  \mathrm{d}\mathfrak{m}(z)= \int_{\mathfrak{h}}^{}\mathfrak{F}(z)\gamma(z)  \mathrm{d}\mu(z)\neq  \int_{\mathfrak{h}}^{}\gamma(z)\mathfrak{F}(z)  \mathrm{d}\mu(z)=\int_{\mathfrak{h}}^{}  \mathrm{d}\mathfrak{m}(z)\mathfrak{F}(z)\;.
\end{gather*}
A detailed study of state-valued measures and their properties is given in
the aforementioned references
\citep{correggi2019arxiv,correggi2020arxiv,falconi2017arxiv}; we will make
extensive use of the results proved in those papers, so the interested reader
shall refer to them.

The last concept needed to understand the main results is that of the
(non-commutative) Fourier transform of a quantum state, and of the Fourier
transform of a state-valued measure. These tools allow to put quantum states
and state-valued measures on the same grounds, to set up the quasi-classical
convergence of the former to the latter.  The \emph{Fourier transform of a
  quantum state} $\Gamma_{\varepsilon}$ is the function
$\hat{\Gamma}_{\varepsilon}:\mathfrak{h}\to \mathfrak{L}^1(\mathscr{H})$
given by
\begin{equation*}
  \hat{\Gamma}_{\varepsilon}(\eta)= \tr_{\mathcal{G}_{\varepsilon}^{\rm s}(\mathfrak{h})} \bigl(\Gamma_{\varepsilon} (W_{\varepsilon}(\eta))\bigr)\;,
\end{equation*}
with $W_{\varepsilon}(\eta)$ being the bosonic Weyl operator
\begin{equation*}
  W_{\varepsilon}(\eta)= e^{i\varphi_{\varepsilon}(\eta)}= e^{i\bigl(a^{*}_{\varepsilon}(\eta)+ a_{\varepsilon}(\eta)\bigr)}\;,
\end{equation*}
and $\tr_{\mathcal{G}_{\varepsilon}^{\rm s}(\mathfrak{h})}$ denoting the
partial trace w.r.t.\ the bosonic degrees of freedom.  The
\emph{Fourier transform of a state-valued measure} $\mathfrak{m}$ is the
function $\hat{\mathfrak{m}}:\mathfrak{h}\to \mathfrak{L}^1(\mathscr{H})$
given by
\begin{equation*}
  \hat{\mathfrak{m}}(\eta)= \int_{\mathfrak{h}}^{}e^{2i\Re \langle \eta  , z \rangle_{\mathfrak{h}}}  \mathrm{d}\mathfrak{m}(z)=\int_{\mathfrak{h}}^{}e^{2i\Re \langle \eta  , z \rangle_{\mathfrak{h}}}  \gamma(z)\mathrm{d}\mu(z)\;.
\end{equation*}
We say that a state $\Gamma_{\varepsilon}$ converges quasi-classically to a
state-valued measure $\mathfrak{m}$, denoted by
$\Gamma_{\varepsilon}\underset{\varepsilon\to
  0}{\longrightarrow}\mathfrak{m}$, if and only if for all $\eta\in
\mathfrak{h}$,
\begin{equation*}
  \wstarlim_{\varepsilon\to 0}\, \hat{\Gamma}_{\varepsilon}(\eta) = \hat{\mathfrak m}(\eta)\;,
\end{equation*}
where $\wstarlim$ stands for the limit in the weak-* topology of
$\mathfrak{L}^1(\mathscr{H})$, \emph{i.e.} when tested with compact operators
$\mathfrak{k}\in \mathfrak{L}^{\infty}(\mathscr{H})$:
$$
\Gamma_\varepsilon \underset{\varepsilon\to 0}{\longrightarrow} \mathfrak
m\quad \overset{\rm def}{\Longleftrightarrow} \quad \tr_{\mathscr{H}}
\bigl(\hat{\Gamma}_\varepsilon(\eta) \mathfrak k\bigr)
\underset{\varepsilon\to 0}{\longrightarrow} \tr_{\mathscr{H}}
\bigl(\hat{\mathfrak m}(\eta) \mathfrak k\bigr),\ \ \forall \eta\in\mathfrak
h, \mathfrak{k}\in \mathfrak{L}^{\infty}(\mathscr{H}).
$$

\begin{proposition}[Quasi-classical convergence {\citep[][Prop.\
    2.3]{correggi2019arxiv}}]
  \label{prop:2}
  \mbox{}   \\
  Let $\Gamma_{\varepsilon}$ be a state such that there exist $\delta, C>0$
  with
  \begin{equation}
    \label{eq:1}
    \tr\Bigl((\mathrm{d}\mathcal{G}_{\varepsilon}(1)+1)^{\delta}\Gamma_{\varepsilon}\Bigr)\leq C\; .
  \end{equation}
  Then there exists a sequence $\varepsilon_n\underset{n\to
    \infty}{\longrightarrow} 0$, and a state-valued measure $\mathfrak{m}$
  (in general depending on the sequence) such that
  \begin{equation*}
    \Gamma_{\varepsilon_n}\underset{n\to \infty}{\longrightarrow}\mathfrak{m}\;.
  \end{equation*}
\end{proposition}
\begin{proof}[Proof (sketch)]
  The proof of this proposition adapts to the quasi-classical setting the
  semiclassical analysis for quantum fields developed by Ammari and Nier in
  \citep{ammari2008ahp,ammari2009jmp,ammari2011jmpa,ammari2015asns}, with
  some crucial differences. A complete proof is given in
  \citep{correggi2019arxiv}, however the key ideas could be summarized as
  follows.

  The Fourier transform of a quantum state enjoys some special properties
  \citep[see][]{segal1959mfmdvs,segal1961cjm}, namely:
  \begin{itemize}
  \item $\tr_{\mathscr{H}}\bigl(\hat{\Gamma}_{\varepsilon}(0)\bigr)=1$;
    
  \item $\hat{\Gamma}_{\varepsilon}$ \emph{is weak-* continuous when
      restricted to any finite-dimensional subspace of $\mathfrak{h}$};
    
  \item $\hat{\Gamma}_{\varepsilon}$ is \emph{``quantum-'' completely
      positive definite}: for any finite collection
    $\{\eta_j\}_{j=1}^J\subset \mathfrak{h}$, and
    $\{\mathfrak{t}_j\}_{j=1}^J\subset \mathscr{B}(\mathscr{H})$,
    \begin{equation*}
      \sum_{j,k=1}^J\mathfrak{t}_j\hat{\Gamma}_{\varepsilon}(\eta_j-\eta_k)\mathfrak{t}_k^{*}e^{i\varepsilon \Im \langle \eta_j  , \eta_k \rangle_{}}\geq 0
    \end{equation*}
    as an operator on $\mathscr{H}$.
  \end{itemize}
  Intuitively, by taking the ($\mathfrak{h}$-pointwise) weak-* limit (using a
  compactness argument), one tries to prove that there exists a sequence
  $\varepsilon_n\to 0$ such that $\hat{\Gamma}_0=\lim_{n\to
    \infty}\hat{\Gamma}_{\varepsilon_{n}}$ satisfies
  \begin{itemize}
  \item $\tr_{\mathscr{H}}\bigl(\hat{\Gamma}_0(0)\bigr)=1$;
    
  \item $\hat{\Gamma}_0$ is \emph{weak-* continuous} when restricted to any
    finite-dimensional subspace of $\mathfrak{h}$;
    
  \item $\hat{\Gamma}_0$ is \emph{completely positive definite}: for any
    finite collection $\{\eta_j\}_{j=1}^J\subset \mathfrak{h}$, and
    $\{\mathfrak{t}_j\}_{j=1}^J\subset \mathscr{B}(\mathscr{H})$,
    \begin{equation*}
      \sum_{j,k=1}^J\mathfrak{t}_j\hat{\Gamma}_0(\eta_j-\eta_k)\mathfrak{t}_k^{*}\geq 0\;.
    \end{equation*}
  \end{itemize}
  It turns out that, under the assumption \eqref{eq:1} above, the second and
  third properties are indeed satisfied, thus by the infinite dimensional
  version of Bochner's theorem \citep{falconi2017arxiv}, $\hat{\Gamma}_0$
  identifies uniquely a \emph{cylindrical state-valued measure}\footnote{A
    cylindrical measure is a finitely additive measure that is
    $\sigma$-additive on any subalgebra of cylinders generated by a finite
    number of vectors.}  $\mathfrak{m}$. In addition, \eqref{eq:1} also
  implies that $\mathfrak{m}$ is tight, and thus a Borel Radon measure. The
  first property, namely that the mass is preserved in the limit, \emph{does
    not hold} in general in the quasi-classical case, contrarily to the
  semiclassical case where it is again ensured by \eqref{eq:1}. This is due
  to the fact that some mass may be lost ``at infinity'' if the spin
  subsystem \verde {has infinitely many degrees of freedom}, see
  \cref{sec:losing-mass-through} for a detailed discussion.
\end{proof}

We are now in a position to state the main result of this note, in an
informal but intuitive manner.
\begin{thm}[Quasi-classical dynamics]
  \label{thm:1}
  \mbox{}   \\
  Let $\Gamma_{\varepsilon}\in \mathfrak{L}^1_{+,1}\bigl(\mathscr{H}\otimes
  \mathcal{G}_{\varepsilon}^{\rm s}(\mathfrak{h})\bigr)$ be such that there
  exists $\delta,C>0$ such that, uniformly w.r.t.\ $\varepsilon\in (0,1)$,
  \begin{equation*}
    \tr\Bigl( \bigl(\mathrm{d}\mathcal{G}_{\varepsilon}(1)+1\bigr)^{\delta} \Gamma_{\varepsilon} \Bigr)\leq C\;.
  \end{equation*}
  Then there is a sequence $\varepsilon_n\rightarrow 0$ such that,
    with $\nu := \lim_{\varepsilon\to 0}\varepsilon\nu(\varepsilon)$, the
  following diagram is commutative, for any $t\in \mathbb{R}$:
 
  \begin{equation*}
    \begin{tikzcd}
      \Gamma_{\varepsilon_n} \arrow[r, "\text{quantum evolution}",mapsto] \arrow[swap,d, "\varepsilon_n \underset{n\to \infty}{\longrightarrow} 0",mapsto] &[2.5cm] \Gamma_{\varepsilon_n}(t) \arrow[d, "\varepsilon_n\underset{n\to \infty}{\longrightarrow}0",mapsto]	\\[1cm]
      \gamma(z) \: \diff \mu(z) \arrow[swap,r, "\text{quasi-classical
        evolution}",mapsto] & \mathfrak{U}_{t,0}(z) \gamma(z)
      \mathfrak{U}^{*}_{t,0}(z) \: \diff \lf( e^{-it\nu\omega}\, _{\star}
      \,\mu \ri)(z).
    \end{tikzcd}
  \end{equation*}
\end{thm}
\medskip

In the above theorem, the symbol $(\,\cdot \,) _{\star}(\,\cdot \,)$ stands
for the pushforward of the measure on the right by means of the map on the
left, and $\mathfrak{U}_{t,s}(z)$ is the two-parameter unitary group on
$\mathscr{H}$ generated by the self-adjoint, generally {\em time-dependent}
effective Hamiltonian\footnote{$\mathfrak{U}_{t,s}(z)$ is the unique solution
  of $i\partial_t\mathfrak{U}_{t,s}(z) = \mathfrak{H}(z)
  \mathfrak{U}_{t,s}(z)$ and $\mathfrak{U}_{t,t}(z)=\mathbf 1$.}
\begin{equation*}
  \mathfrak{H}(z)= \mathfrak{S} + 2\Re \langle e^{-it\nu\omega}z  , g \rangle_{\mathfrak{h}} \, \mathfrak{s}\;.
\end{equation*}
The operator $\mathfrak{H}(z)$ is a time-dependent generator if $\nu=1$, and
it is time-independent if $\nu=0$\footnote{We restrict our attention only to
  the limits $\nu=1$ and $\nu=0$, since they encode all different and
  well-defined outcomes that one could obtain for the effective dynamics. In
  fact, every choice of $\nu(\varepsilon)$ such that $\lim_{\varepsilon\to
    0}\varepsilon \nu(\varepsilon)= \lambda >0$ would amount in a rescaling
  of the field dispersion relation, while any choice such that either
  $\lim_{\varepsilon\to 0}\varepsilon \nu(\varepsilon)= \lambda = \infty$ or
  such that the limit does not exist would prevent an explicit definition of
  the effective dynamics in the limit $\varepsilon\to 0$}. More precisely,
for $\nu=1$ the classical bosonic field described by $e^{-it\omega}\,
_{\star} \,\mu$ evolves freely, while for $\nu=0$ it does not evolve at all
and is described by $\mu$ at all times; in both cases it drives the spin
state through $\mathfrak{U}_{t,0}(z)$, mediated over all possible
configurations $z$ in the support of $\mu$. Let us also remark that
$\mathfrak{U}_{t,0}(z) \gamma(z) \mathfrak{U}^{*}_{t,0}(z)$ shall be seen as
a Radon-Nikod\'{y}m derivative, and as such the pushforward does not act on
it.  \smallskip

\cref{thm:1} therefore explains how the semiclassical bosonic
subsystem becomes an environment driving the spin system, unaffected by the
latter, if the quasi-classical parameter $\varepsilon$ is small
enough. This also motivates the terminology used so far, {\it i.e.,
  the identification of the spin component as the `small' system, while the
  bosonic field is the `large' environment or reservoir.}  In addition, the
effective dynamics of the spin system can be characterized explicitly, being
unitary and described by $\mathfrak{U}_{t,0}(z)$ for any fixed configuration
$z$ of the classical field, but not unitary (and not even
Markovian\footnote{We plan to investigate the non-Markovian character of the
  quasi-classical effective dynamics in an upcoming paper.}) in general, due
to the integration over all configurations reached by the classical
bosonic state $e^{-it\nu\omega}\, _{\star} \,\mu$. Both a stationary
and a freely evolving bosonic environment can be obtained, tuning the
microscopic initial state accordingly in a way that makes $\nu(\varepsilon)$
either $1$ (stationary) or $\frac{1}{\varepsilon}$ (freely evolving). Let us
stress that even if $\nu(\varepsilon)$ appears in the Hamiltonian, it should
be thought as a feature of the chosen initial state, fixing the scale of
energy for the bosonic subsystem.

\subsection{Loss of mass in the quasi-classical limit}
\label{sec:losing-mass-through}

An interesting feature of quasi-classical systems is that some mass can be
lost in the limit $\varepsilon\to 0$, due to the entanglement between the two
subsystems, when the spin part is infinite dimensional. By loss of
mass we mean that the measure in the quasi-classical limit satisfies
$\mu(\mathfrak h)<1$. There are many well-known examples of loss of mass
(also called loss of compactness) in semiclassical analysis, both in finite
and infinite dimensions \citep[see,
\emph{e.g.},][]{ammari2008ahp,lions1993rmi}. In those cases, however,
conditions like \eqref{eq:1} are enough to guarantee that no mass is lost.

Here, on the contrary, mass can be lost ``through the spin system'',
provided the systems are entangled, and the spin system components could
escape to infinity. In fact, if the microscopic state is unentangled (in a
natural quasi-classical way), \emph{i.e.} it is of the form
\begin{gather*}
  \Gamma_{\varepsilon}= \gamma_0\otimes \xi_{\varepsilon}\;, \\\text{with } \gamma_0\in \mathfrak{L}^1_{+,1}\bigl(\mathscr{H}\bigr)\;,\; \xi_{\varepsilon}\in \mathfrak{L}^1_{+,1}\bigl(\mathcal{G}_{\varepsilon}^{\rm s}(\mathfrak{h})\bigr)\;,
\end{gather*}
the quasi-classical convergence in \cref{prop:2} ``decouples'' and no mass
can be lost: for this class of states \eqref{eq:1} implies
$\mathfrak{m}=(\mu,\gamma_0)$, with $\mu(\mathfrak{h})=1$. Similarly,
if the spin subsystem is finite dimensional or its particles are
confined, again no mass can be lost. More precisely, if either
$\dim(\mathscr{H})<+\infty $ or there exists an operator $\mathfrak{A}>0$ on
$\mathscr{H}$ with compact resolvent\footnote{If $\mathfrak{S}$ has compact
  resolvent (and it is bounded from below), $\mathfrak{A}=\mathfrak{S} +
  \lvert \inf \sigma(\mathfrak{S}) \rvert_{}^{}+1$ would be a natural choice,
  and the associated condition \eqref{eq:2} would mean that mass is not lost
  if one restricts to states with $\varepsilon$-uniformly-bounded Spin
  kinetic energy.} such that there exists $C>0$ with
\begin{equation}
  \label{eq:2}
  \tr\bigl(\Gamma_{\varepsilon}( \mathfrak{A}\otimes 1) \bigr)\leq C\;,
\end{equation}
then the measure $\mathfrak{m}=(\mu,\gamma)$ in \cref{prop:2} is such that
$\mu(\mathfrak{h})=1$.

In general however, part or all of the mass can be lost in the limit
$\varepsilon\to 0$ of a generic quantum state
$\Gamma_{\varepsilon}$. \cref{thm:1} is also interesting if (some) mass is
lost. In fact, \emph{the mass is preserved by the quasi-classical dynamics}:
this means that the same amount of mass is lost at any time, and therefore
that one should check if any mass is lost only at the initial time. We think
that this loss of mass phenomenon peculiar to the quasi-classical
entanglement is worth pointing out, and could be explored further in concrete
applications.

\section{Heuristic Derivation}
\label{sec:an-heur-argum}

If the initial state is quasi-classically unentangled, \emph{i.e.}
\begin{equation*}
  \Gamma_{\varepsilon}=\gamma_0\otimes \xi_{\varepsilon}
\end{equation*}
(see \cref{sec:losing-mass-through} above), it is possible to use the
factorized nature of the spin-boson interaction to formally obtain a result
akin to \cref{thm:1} in a very intuitive way, that hopefully helps to
illustrate the main ideas behind the general proof.


In order to discuss the strategy, let us set some useful notation. Define the
free Hamiltonian
\begin{equation*}
  H_{\varepsilon}^{\rm f}:= H_{\varepsilon}\bigr\rvert_{g=0}= \mathfrak{S}\otimes 1 + \nu(\varepsilon)\, 1\otimes \mathrm{d}\mathcal{G}_{\varepsilon}(\omega)=: H^{\rm fs}+ H_{\varepsilon}^{\rm fb}\;,
\end{equation*}
and define the interaction
\begin{equation*}
  H_{\varepsilon}^{\rm i}= H_{\varepsilon}- H_{\varepsilon}^{\rm f}\;.
\end{equation*}
The Dyson expansion for the evolution in the interaction picture is
\begin{equation*}
  e^{itH_{\varepsilon}^{\rm f}}e^{-it H_{\varepsilon}}= 1+ \sum_{n\in \mathbb{N}_{*}}^{}(-i)^n \int_0^{t}  \mathrm{d}s_1\int_0^{s_1}  \mathrm{d}s_2 \dotsm \int_0^{s_{n-1}}  \mathrm{d}s_n \;\mathfrak{s}_{s_1}\dotsm \mathfrak{s}_{s_n} \otimes \varphi_{\varepsilon,s_1}\dotsm\varphi_{\varepsilon,s_n}\;,
\end{equation*}
where
\begin{gather*}
  \mathfrak{s}_s= e^{is H^{\rm fs}}\mathfrak{s} e^{-isH^{\rm fs}}\;,\\
  \varphi_{\varepsilon,s}= e^{is H^{\rm fb}_{\varepsilon}}\varphi_{\varepsilon}(g)e^{-is H^{\rm fb}_{\varepsilon}}=\varphi_{\varepsilon}(e^{i\nu(\varepsilon)s \omega }g)\;;
\end{gather*}
and in addition
\begin{equation*}
  e^{it H_{\varepsilon}}e^{-itH_{\varepsilon}^{\rm f}}= 1+ \sum_{m\in \mathbb{N}_{*}}^{}i^m \int_0^{t}  \mathrm{d}u_1\int_0^{u_1}  \mathrm{d}u_2 \dotsm \int_0^{u_{m-1}}  \mathrm{d}u_m \;\mathfrak{s}_{u_m}\dotsm \mathfrak{s}_{u_1} \otimes \varphi_{\varepsilon,u_m}\dotsm\varphi_{\varepsilon, u_1}\;.
\end{equation*}
It follows that
\begin{multline*}
  \tilde{\gamma}_\varepsilon(t):=\tr_{\mathcal{G}_{\varepsilon}^{\rm s}(\mathfrak{h})}\Bigl(e^{it H^{\rm fs}}e^{-it H_{\varepsilon}}(\gamma_0\otimes \xi_{\varepsilon})e^{it H_{\varepsilon}}e^{-it H^{\rm fs}}\Bigr)\\=\sum_{m,n\in \mathbb{N}}^{}i^{m-n} \int_0^{t}  \mathrm{d}s_1 \dotsm \int_0^{s_{n-1}} \mspace{-17mu} \mathrm{d}s_n\int_0^{t}  \mathrm{d}u_1 \dotsm \int_0^{u_{m-1}} \mspace{-17mu} \mathrm{d}u_m \;\mathfrak{s}_{s_1}\dotsm \mathfrak{s}_{s_n}\,\gamma_0\, \mathfrak{s}_{u_m}\dotsm \mathfrak{s}_{u_1} \\\tr_{\mathcal{G}_{\varepsilon}^{\rm s}(\mathfrak{h})}\Bigl( \xi_{\varepsilon}\varphi_{\varepsilon,u_m}\dotsm\varphi_{\varepsilon,u_1}\varphi_{\varepsilon,s_1}\dotsm\varphi_{\varepsilon,s_n}\Bigr)\;.
\end{multline*}
Now, in order to take the limit $\varepsilon\to 0$, one should focus on the
expectation with respect to $\xi_{\varepsilon}$:
\begin{equation*}
  \langle \varphi_{\varepsilon,u_m}\dotsm\varphi_{\varepsilon,u_1}\varphi_{\varepsilon,s_1}\dotsm\varphi_{\varepsilon,s_n}  \rangle_{\xi_{\varepsilon}}:= \tr_{\mathcal{G}_{\varepsilon}^{\rm s}(\mathfrak{h})}\Bigl( \xi_{\varepsilon}\varphi_{\varepsilon,u_m}\dotsm\varphi_{\varepsilon,u_1}\varphi_{\varepsilon,s_1}\dotsm\varphi_{\varepsilon,s_n}\Bigr)\;.
\end{equation*}
It is possible to write such an expectation as follows, where
$f_1,\ldots,f_k\in\mathfrak h$,
\begin{multline*}
  \langle \varphi_{\varepsilon}(f_1)\dotsm\varphi_{\varepsilon}(f_k)  \rangle_{\xi_{\varepsilon}}= (-i)^k\partial_{\lambda_1}\dotsm \partial_{\lambda_k} \Bigl(\langle W_{\varepsilon}(\lambda_1f_1)\dotsm W_{\varepsilon}(\lambda_kf_k)  \rangle_{\xi_{\varepsilon}}\Bigr)\biggr\rvert_{\lambda_1=\dotsm=\lambda_k=0}\\=: D_k\langle W_{\varepsilon}(\lambda_1f_1)\dotsm W_{\varepsilon}(\lambda_kf_k)  \rangle_{\xi_{\varepsilon}}\biggr\rvert_{\underline{\lambda}=0} \; .
\end{multline*}
It then follows from the Weyl CCR
\begin{equation*}
  W_{\varepsilon}(\lambda_1f_1)W_{\varepsilon}(\lambda_2f_2)= e^{-i\varepsilon\Im \langle \lambda_1f_1  , \lambda_2f_2 \rangle_{}}W_{\varepsilon}(\lambda_1f_1+\lambda_2f_2)
\end{equation*}
that
\begin{equation*}
  \lim_{\varepsilon\to 0}D_k\langle W_{\varepsilon}(\lambda_1f_1)\dotsm W_{\varepsilon}(\lambda_kf_k)  \rangle_{\xi_{\varepsilon}}\biggr\rvert_{\underline{\lambda}=0}= \lim_{\varepsilon\to 0}D_k\langle W_{\varepsilon}(\lambda_1f_1+\dotsm+\lambda_kf_k)  \rangle_{\xi_{\varepsilon}}\biggr\rvert_{\underline{\lambda}=0}\;.
\end{equation*}
Now, in this formal reasoning we feel free to exchange $\lim_{\varepsilon\to
  0}$ with $D_k$; thus we obtain, provided that\footnote{The scalar
  convergence $\xi_{\varepsilon}\underset{\varepsilon\to 0}{\longrightarrow}
  \mu$ is perfectly analogous to the quasi-classical one, and could be seen
  as a particular case of it where the additional degrees of freedom are
  trivial. Let us remark again that for the scalar case -- and thus also for
  the unentangled quasi-classical states considered here -- \eqref{eq:1} is
  sufficient to guarantee that $\mu(\mathfrak{h})=1$.}
$\xi_{\varepsilon}\underset{\varepsilon\to 0}{\longrightarrow} \mu$,
\begin{equation*}
  D_k\lim_{\varepsilon\to 0}\langle W_{\varepsilon}(\lambda_1f_1+\dotsm+\lambda_kf_k)  \rangle_{\xi_{\varepsilon}}\biggr\rvert_{\underline{\lambda}=0}= D_k \hat{\mu}(\lambda_1f_1+\dotsm \lambda_kf_k)\biggr\rvert_{\underline{\lambda}=0}= \int_{\mathfrak{h}}^{}\alpha_{f_1}(z)\dotsm \alpha_{f_k}(z)  \mathrm{d}\mu(z)\;;
\end{equation*}
where
\begin{equation*}
  \alpha_f(z):= 2\Re \langle z  , f \rangle_{\mathfrak{h}}\;.
\end{equation*}
Applying these results to $\lim_{\varepsilon\to
  0}\tilde{\gamma}_{\varepsilon}(t)$ yields:
\begin{multline*}
  \lim_{\varepsilon\to 0}\tilde{\gamma}_\varepsilon(t)=\sum_{m,n\in \mathbb{N}}^{}i^{m-n} \int_0^{t}  \mathrm{d}s_1 \dotsm \int_0^{s_{n-1}} \mspace{-17mu} \mathrm{d}s_n\int_0^{t}  \mathrm{d}u_1 \dotsm \int_0^{u_{m-1}} \mspace{-17mu} \mathrm{d}u_m \;\mathfrak{s}_{s_1}\dotsm \mathfrak{s}_{s_n}\,\gamma\, \mathfrak{s}_{u_m}\dotsm \mathfrak{s}_{u_1} \\\int_{\mathfrak{h}}^{}\alpha_{s_1}(z)\dotsm \alpha_{s_n}(z)\alpha_{u_1}(z)\dotsm\alpha_{u_m}(z)  \mathrm{d}\mu(z)\;,
\end{multline*}
with
\begin{equation}
  \label{eq:4}
  \alpha_s(z)=2\Re \langle z  , e^{i\nu s \omega}g \rangle_{\mathfrak{h}}\;.
\end{equation}
We have thus
\begin{multline}
  \gamma(t):=\lim_{\varepsilon\to 0}\gamma_{\varepsilon}(t):= \lim_{\varepsilon\to 0}\tr_{\mathcal{G}_{\varepsilon}^{\rm s}(\mathfrak{h})}\Bigl(e^{-it H_{\varepsilon}}(\gamma_0\otimes \xi_{\varepsilon})e^{it H_{\varepsilon}}\Bigr)\\
  =\int_{\mathfrak{h}}^{} \mathfrak{U}_{t,0}(z)\gamma(z)\mathfrak{U}^{*}_{t,0}(z)  \mathrm{d}\bigl(e^{-it\nu \omega}\,_{\star}\, \mu\bigr)(z)\;,
  \label{eq:4.1}
\end{multline}
where $\mathfrak{U}_{t,0}(z)$ is defined in \cref{thm:1}, that can also
formally be seen as
\begin{equation*}
  \mathfrak{U}_{t,0}(z)= e^{-itH^{\rm fs}}\sum_{n\in \mathbb{N}}^{}(-i)^n\int_0^{t}  \mathrm{d}s_1 \dotsm \int_0^{s_{n-1}} \mspace{-17mu} \mathrm{d}s_n\;\mathfrak{s}_{s_1}\alpha_{s_1}(z) \dotsm \mathfrak{s}_{s_n}\alpha_{s_n}(z)\;.
\end{equation*}
One can see last equality in \eqref{eq:4.1} as a `resummation of the Dyson
series'.

\section{Proof of \cref{thm:1}}
\label{sec:proof-cref}

As we have seen in \cref{sec:an-heur-argum}, the factorized structure of the
spin-boson interaction can be used to simplify the study of the
quasi-classical limit, compared to, say, the Nelson, polaron, or Pauli-Fierz
models, where such a factorization is not present \citep[see][for their
quasi-classical analysis]{correggi2019arxiv,correggi2020arxiv}. The proof of
\cref{thm:1} reflects this as well, as illustrated below.  Since the proof
follows closely \citep{correggi2019arxiv} -- and directly utilizes some of
its results -- we will mostly focus on highlighting the features specific to
the spin-boson model.

The proof is organized in a few steps, namely:
\begin{itemize}
\item write a Duhamel-type formula for the Fourier transform of evolved
  quantum states in the interaction representation;
  
\item extract a subsequence $\varepsilon_{n_k}$ of common quasi-classical
  convergence for regular enough evolved states at any given time;
  
\item take the limit $\varepsilon_{n_k}\to 0$ along the aforementioned
  subsequence of the Duhamel formula;
  
\item study the resulting transport equation to identify the evolved measure,
  and uniqueness of the limit;
  
\item relax the regularity assumption needed at step two to the assumption in
  the theorem.
\end{itemize}
We will review these steps below separately.

\subsection{The Duhamel Formula}
\label{sec:duhamel-formula}

For technical reasons, related mostly to the possible unboundedness of
$\omega$, it is convenient to pass to the so-called interaction
representation. Let us define the evolution in the interaction representation
as
\begin{equation*}
  \Upsilon_{\varepsilon}(t):= e^{it H_{\varepsilon}^{\rm f}}\Gamma_{\varepsilon}(t)e^{-it H_{\varepsilon}^{\rm f}}\;.
\end{equation*}

The Schrödinger differential equation of quantum evolution requires too much
regularity for its solutions; it is more convenient to use its integral (or
Duhamel) form. To write it, it is sufficient to suppose that for all $t\in
\mathbb{R}$,
$\tr\Bigl(\Upsilon_{\varepsilon}(t)(\mathrm{d}\mathcal{G}_{\varepsilon}(1)+1)^{1/2}\Bigr)<+\infty$. Under
this assumption the Fourier transform $\hat{\Upsilon}_{\varepsilon}(t)$
satisfies the following integral equation, weakly on
$\mathfrak{L}^1(\mathscr{H})$, for any $t,s\in \mathbb{R}$ and $\eta\in
\mathfrak{h}$:
\begin{multline}
  \label{eq:3}
  [\hat{\Upsilon}_{\varepsilon}(t)](\eta)- [\hat{\Upsilon}_{\varepsilon}(s)](\eta)=  -i \int_s^t \tr_{\mathcal{G}_{\varepsilon}^{\rm s}(\mathfrak{h})}\Bigl(\Bigl[\mathfrak{s}(\tau) \otimes \varphi_{\varepsilon}(\tau), \Upsilon_{\varepsilon}(\tau)\Bigr]W_{\varepsilon}(\eta)\Bigr)  \mathrm{d}\tau\\=i \int_s^t \biggl(\tr_{\mathcal{G}_{\varepsilon}^{\rm s}(\mathfrak{h})}\Bigl(  \Upsilon_{\varepsilon}(\tau)\varphi_{\varepsilon}(\tau)W_{\varepsilon}(\eta)\Bigr)\mathfrak{s}(\tau)-\mathfrak{s}(\tau)\tr_{\mathcal{G}_{\varepsilon}^{\rm s}(\mathfrak{h})}\Bigl( \varphi_{\varepsilon}(\tau) \Upsilon_{\varepsilon}(\tau)W_{\varepsilon}(\eta)\Bigr)\biggr)  \mathrm{d}\tau\;.
\end{multline}
Here, we write
$$
\mathfrak s(\tau)=e^{i\tau H^{\rm fs}}\mathfrak s e^{-i\tau H^{\rm fs}}\quad
\mbox{ and}\quad \varphi_\varepsilon(\tau) = e^{i\tau H_\varepsilon^{\rm
    fb}}\varphi_\varepsilon(g)e^{-i\tau H_\varepsilon^{\rm fb}}.
$$
Once the required regularity is taken care of, this equation follows directly
from the algebraic properties of the quantum evolution (in interaction
representation) $e^{it H_{\varepsilon}^{\rm f}}e^{-it H_{\varepsilon}}$, as
already outlined in \cref{sec:an-heur-argum}. The Duhamel formula is the
starting point for our study of the dynamical quasi-classical limit.

The regularity bound concerning the average of the number operator at all
times that we used above -- especially in its form that is uniform
w.r.t. $\varepsilon\in (0,1)$ -- will be crucial also in what follows, so let
us formulate it as an auxiliary ``black box'' result. Such propagation
results are typically heavily dependent on the model under consideration; for
the Spin-Boson system one could adapt very easily the results available for
the Nelson model with ultraviolet cutoff \citep[][Proposition
4.2]{falconi2013jmp}, obtaining the lemma below.
\begin{lemma}
  \label{lemma:1}
  For any $\delta,C>0$ and for all $t\in \mathbb{R}$, there exists
  $K(\delta,C,t)>0$ such that
  \begin{multline*}
    \tr\Bigl(\Gamma_{\varepsilon}\bigl(\mathrm{d}\mathcal{G}_{\varepsilon}(1)+1\bigr)^{\delta}\Bigr)\leq C\; \Longrightarrow \; \biggl(\;\tr\Bigl(\Gamma_{\varepsilon}(t)\bigl(\mathrm{d}\mathcal{G}_{\varepsilon}(1)+1\bigr)^{\delta}\Bigr)\leq K(\delta,C,t)\\\land  \; \tr\Bigl(\Upsilon_{\varepsilon}(t)\bigl(\mathrm{d}\mathcal{G}_{\varepsilon}(1)+1\bigr)^{\delta}\Bigr)\leq K(\delta,C,t)\; \biggr)\;.
  \end{multline*}
\end{lemma}

\subsection{Common subsequence extraction at all times}
\label{sec:comm-subs-extr}

Thanks to the propagation lemma, \cref{lemma:1}, it is possible to prove that
$t\mapsto \hat{\Upsilon}_{\varepsilon}(t)$ is \emph{uniformly equicontinuous
  w.r.t.\ $\varepsilon\in (0,1)$}. This in turn implies, by a diagonal
extraction argument, that starting from any sequence $\varepsilon_n\to 0$, it
is possible to extract a subsequence $\varepsilon_{n_k}\to 0$ that guarantees
convergence of $\Upsilon_{\varepsilon_{n_k}}(\tau)$ to some state-valued
measure $\mathfrak{n}_\tau$ for any $\tau$ in a given compact interval
$[s,t]$ (actually for any given time). This is the crucial ingredient
allowing to study the limit $\varepsilon\to 0$ of the Duhamel formula
\eqref{eq:3}, for the latter involves the integral over all evolved states
between $s$ and $t$. The result reads as follows, and it has been proved in
\citep[][Propositions 4.2 and 4.3]{correggi2019arxiv}, with a general
argument that does not depend on the nature of $\mathscr{H}$ or on the
Hamiltonian (one only requires that a form of \cref{lemma:1} is available).

\begin{proposition}
  \label{prop:3}
  Let $\Gamma_{\varepsilon}$ be such that
  \begin{equation*}
    \tr\Bigl(\Gamma_{\varepsilon}\bigl(\mathrm{d}\mathcal{G}_{\varepsilon}(1)+1\bigr)^{1/2}\Bigr)\leq C\;.
  \end{equation*}
  Then $\mathbb{R}\times \mathfrak{h}\ni
  (t,\eta)\mapsto[\hat{\Upsilon}_{\varepsilon}(t)](\eta)\in
  \mathfrak{L}^1(\mathscr{H})$ is uniformly equicontinuous w.r.t.\
  $\varepsilon\in (0,1)$ on bounded subsets of $\mathbb{R}\times
  \mathfrak{h}$, having endowed $\mathfrak{L}^1(\mathscr{H})$ with the weak-*
  topology.

  In addition, for any sequence $\varepsilon_n\to 0$, there exists a
  subsequence $\varepsilon_{n_k}\to 0$ and a family $\{\mathfrak{n}_t\}_{t\in
    \mathbb{R}}$ of state-valued measures such that for all $t\in
  \mathbb{R}$,
  \begin{equation*}
    \Upsilon_{\varepsilon_{n_k}}(t)\underset{k\to \infty}{\longrightarrow} \mathfrak{n}_t\;.
  \end{equation*}
\end{proposition}
As a byproduct (again this is a general result concerning unitary evolutions
generated by operators of the type
$\nu(\varepsilon)\mathrm{d}\mathcal{G}_{\varepsilon}(\cdot)$), we also get
the following information on the limit of the ``true'' evolution
$\Gamma_{\varepsilon}(t)$. Remember that we defined $\nu=\lim_{\varepsilon\to
  0} \varepsilon\nu(\varepsilon)$.
\begin{corollary}
  \label{cor:1}
  Under the same assumptions as in \cref{prop:3}, and given the subsequence
  $\varepsilon_{n_k}\to 0$ and measures $\{\mathfrak{n}_t\}_{t\in
    \mathbb{R}}$, we have that for any $t\in \mathbb{R}$,
  \begin{equation*}
    \Gamma_{\varepsilon_{n_k}}(t)\underset{k\to \infty}{\longrightarrow} \mathfrak{m}_t= e^{-it H^{\rm fs}}\bigl(e^{-it\nu\omega}\,_{\star}\, \mathfrak{n}_t\bigr)e^{it H^{\rm fs}} \;.
  \end{equation*}
\end{corollary}
In other words, we are able to relate the quasi-classical evolution in the
interaction picture to the one not in interaction picture ``as it should
be'', \emph{i.e.}  by acting with the expected free evolution on both the
Spin and classical Boson subsystems. It follows that once we have
characterized the map $t\to \mathfrak{n}_t$, we have also a characterization
for the map $t\to \mathfrak{m}_t$.

\subsection{The limit of the Duhamel formula}
\label{sec:limit-duham-form}

We are now in a position to take the limit $\varepsilon\to 0$ of the Duhamel
formula \eqref{eq:3}. In taking this step, the factorized nature of the
spin-boson interaction helps greatly, essentially allowing to transform the
problem from quasi-classical to semiclassical, allowing us to avoid
completely the use the heavy machinery of quasi-classical calculus developed
in \citep[][\textsection2]{correggi2019arxiv} (that is however
\emph{necessary} whenever the interaction is not factorized as for the spin
boson). Let $\mathfrak{k}\in \mathfrak{L}^{\infty}(\mathscr{H})$ be a compact
operator on the Spin subsystem, then Duhamel's formula \eqref{eq:3} becomes
\begin{multline*}
  \tr_{\mathscr{H}}\Bigl([\hat{\Upsilon}_{\varepsilon}(t)](\eta)\mathfrak{k}\Bigr)- \tr_{\mathscr{H}}\Bigl([\hat{\Upsilon}_{\varepsilon}(s)](\eta)\mathfrak{k}\Bigr)\\= i \int_s^t \biggl(\tr\Bigl(  \Upsilon_{\varepsilon}(\tau)\varphi_{\varepsilon}(\tau)W_{\varepsilon}(\eta)\mathfrak{s}(\tau)\mathfrak{k}\Bigr)-\tr\Bigl(\mathfrak{k} \mathfrak{s}(\tau) 
  \varphi_{\varepsilon}(\tau) \Upsilon_{\varepsilon}(\tau)W_{\varepsilon}(\eta)\Bigr)\biggr)  \mathrm{d}\tau\;.
\end{multline*}
It is possible to exchange the trace w.r.t.\ $\mathscr{H}$ and the integral
by dominated convergence, using the bound for $\Gamma_{\varepsilon}$ assumed
in \cref{prop:3}, and its time propagation given by \cref{lemma:1}. By
\cref{prop:3}, and the definition of quasi-classical convergence, it follows
immediately that, along the common subsequence $\varepsilon_{n_k}\to 0$,
\begin{gather*}
  \lim_{k\to \infty}\tr_{\mathscr{H}}\Bigl([\hat{\Upsilon}_{\varepsilon_{n_k}}(t)](\eta)\mathfrak{k}\Bigr)= \tr_{\mathscr{H}}\Bigl(\hat{\mathfrak{n}}_t(\eta)\mathfrak{k}\Bigr)\;,\\
  \lim_{k\to \infty}\tr_{\mathscr{H}}\Bigl([\hat{\Upsilon}_{\varepsilon_{n_k}}(s)](\eta)\mathfrak{k}\Bigr)= \tr_{\mathscr{H}}\Bigl(\hat{\mathfrak{n}}_s(\eta)\mathfrak{k}\Bigr)\;.
\end{gather*}
Let us now focus on the interaction term, and in particular on the expression
\begin{equation*}
  \tr\Bigl(  \Upsilon_{\varepsilon}(\tau)\varphi_{\varepsilon}(\tau)W_{\varepsilon}(\eta)\mathfrak{s}(\tau)\mathfrak{k}\Bigr)\;,
\end{equation*}
the other one being analogous. Let us now decompose the operator
$\mathfrak{s}(\tau)\mathfrak{k}$ in its real positive, negative, and
imaginary positive, negative parts:
\begin{equation*}
  \mathfrak{s}(\tau)\mathfrak{k}= \mathfrak{sk}_{r+}-\mathfrak{sk}_{r-}+i(\mathfrak{sk}_{i+}-\mathfrak{sk}_{i-})\;,
\end{equation*}
with
$\mathfrak{sk}_{r+},\mathfrak{sk}_{r-},\mathfrak{sk}_{i+},\mathfrak{sk}_{i-}\geq
0$. Therefore, we have that
\begin{multline*}
  \tr\Bigl(  \Upsilon_{\varepsilon}(\tau)\varphi_{\varepsilon}(\tau)W_{\varepsilon}(\eta)\mathfrak{s}(\tau)\mathfrak{k}\Bigr)= \tr\Bigl(  \Upsilon_{\varepsilon}(\tau)\varphi_{\varepsilon}(\tau)W_{\varepsilon}(\eta)\mathfrak{sk}_{r+}\Bigr)- \tr\Bigl(  \Upsilon_{\varepsilon}(\tau)\varphi_{\varepsilon}(\tau)W_{\varepsilon}(\eta)\mathfrak{sk}_{r-}\Bigr)\\+i\biggl(\tr\Bigl(  \Upsilon_{\varepsilon}(\tau)\varphi_{\varepsilon}(\tau)W_{\varepsilon}(\eta)\mathfrak{sk}_{i+}\Bigr)- \tr\Bigl(  \Upsilon_{\varepsilon}(\tau)\varphi_{\varepsilon}(\tau)W_{\varepsilon}(\eta)\mathfrak{sk}_{i-}\Bigr)\biggr)\;.
\end{multline*}
Now, we would like to treat all these terms in the same fashion, so let us
focus on the first one. We can split the total trace in the two partial
traces, but we do it\emph{ in reverse order w.r.t.\ before}:
\begin{equation*}
  \tr\Bigl(  \Upsilon_{\varepsilon}(\tau)\varphi_{\varepsilon}(\tau)W_{\varepsilon}(\eta)\mathfrak{sk}_{r+}\Bigr)=\tr_{\mathcal{G}_{\varepsilon}^{\rm s}(\mathfrak{h})}\biggl(  \tr_{\mathscr{H}}\Bigl(\Upsilon_{\varepsilon}(\tau)\mathfrak{sk}_{r+}\Bigr) \varphi_{\varepsilon}(\tau)W_{\varepsilon}(\eta)\biggr)\;.
\end{equation*}
The partial trace w.r.t.\ to $\mathscr{H}$ is the expectation over a state of
a positive operator, so
\begin{equation*}
  \zeta_{\varepsilon}(\tau,\mathfrak{sk}_{r+}):= \tr_{\mathscr{H}}\Bigl(\Upsilon_{\varepsilon}(\tau)\mathfrak{sk}_{r+}\Bigr)\in \mathfrak{L}^1_+(\mathfrak{h})\;,
\end{equation*}
and we finally obtain
\begin{equation*}
  \tr\Bigl(  \Upsilon_{\varepsilon}(\tau)\varphi_{\varepsilon}(\tau)W_{\varepsilon}(\eta)\mathfrak{sk}_{r+}\Bigr)=\tr_{\mathcal{G}_{\varepsilon}^{\rm s}(\mathfrak{h})}\biggl(  \zeta_{\varepsilon}(\tau,\mathfrak{sk}_{r+}) \varphi_{\varepsilon}(\tau)W_{\varepsilon}(\eta)\biggr)\;.
\end{equation*}
The state $\zeta_{\varepsilon}(\tau,\mathfrak{sk}_{r+})$ is a semiclassical
(scalar) state, living on the Fock space. On one hand, by \cref{prop:3} and
the definition of quasi-classical convergence\footnote{Quasi-classical
  convergence is the pointwise convergence of Fourier transforms in weak-*
  topology, \emph{i.e.}  when tested with compact operators. Since
  $\mathfrak{sk}_{r+}$ is compact, we have pointwise convergence of
  $\hat{\Upsilon}_{\varepsilon}(\tau)$ traced together with
  $\mathfrak{sk}_{r+}$.} we know that
\begin{equation*}
  \zeta_{\varepsilon_{n_k}}(\tau,\mathfrak{sk}_{r+})\underset{k\to \infty}{\longrightarrow} \mathrm{d}\mu_{\tau,\mathfrak{sk}_{r+}}(z)= \tr_{\mathscr{H}}\bigl(\mathrm{d}\mathfrak{n}_{\tau}(z) \mathfrak{sk}_{r+}\bigr)\;.
\end{equation*}
On the other hand, by semiclassical calculus in infinite dimensions
\citep[see][]{ammari2008ahp}, we also know that
\begin{equation*}
  \lim_{k\to \infty}\tr_{\mathcal{G}_{\varepsilon_{n_k}}^{\rm s}(\mathfrak{h})}\biggl(  \zeta_{\varepsilon_{n_k}}(\tau,\mathfrak{sk}_{r+}) \varphi_{\varepsilon_{n_k}}(\tau)W_{\varepsilon_{n_k}}(\eta)\biggr)= \int_{\mathfrak{h}}^{}\alpha_{\tau}(z)e^{2i\Re \langle \eta  , z \rangle_{\mathfrak{h}}}  \mathrm{d}\mu_{\tau,\mathfrak{sk}_{r+}}(z)\;,
\end{equation*}
where the shorthand $\alpha_{\tau}(z)$ has been defined in
\eqref{eq:4}. Combining the two things, and repeating the same reasoning for
all the other remaining terms, we end up obtaining the following integral
equation for the map $t\to \mathfrak{n}_t$ (another dominated convergence
argument allows to pass the limit $\varepsilon_{n_k}\to 0$ inside the time
integral, this time exploiting the uniformity w.r.t.\ $\varepsilon\in (0,1)$
of the number operator bounds at any time).
\begin{proposition}
  \label{prop:4}
  The family of state-valued measures $\{\mathfrak{n}_t\}_{t\in \mathbb{R}}$
  of \cref{prop:3} satisfies the following transport equation for the Fourier
  transform, in the weak sense on $\mathfrak{L}^1(\mathscr{H})$:
  \begin{equation*}
    \hat{\mathfrak{n}}_t(\eta)-\hat{\mathfrak{n}}_s(\eta)=i \int_s^t\int_{\mathfrak{h}}^{}[\gamma_{\mathfrak{n}_{\tau}}(z),\mathfrak{s}(\tau)]  \alpha_{\tau}(z) e^{2i\Re \langle \eta  , z \rangle_{\mathfrak{h}}}\mathrm{d}\mu_{\mathfrak{n}_{\tau}}(z)  \mathrm{d}\tau\;.
  \end{equation*}
\end{proposition}

\subsection{Uniqueness of the solution to the transport equation, uniqueness
  of the limit}
\label{sec:uniq-solut-transp}

The transport equation for the Fourier transform of $\mathfrak{n}_t$ can be
easily translated in an equation for the measure:
\begin{equation*}
  \gamma_{\mathfrak{n}_t}(z)\mathrm{d}\mu_{\mathfrak{n}_t}(z) - \gamma_{\mathfrak{n}_s}(z)\mathrm{d}\mu_{\mathfrak{n}_s}(z)= i\int_s^t  [\gamma_{\mathfrak{n}_{\tau}}(z),\mathfrak{s}(\tau)]  \alpha_{\tau}(z) \mathrm{d}\mu_{\mathfrak{n}_{\tau}}(z)\mathrm{d}\tau\;.
\end{equation*}
Now, let us fix $s=0$, and suppose that we have the quasi-classical
convergence at initial time
\begin{equation*}
  \Gamma_{\varepsilon_n}\underset{n\to \infty}{\longrightarrow} \mathfrak{m}\;.
\end{equation*}
It then follows that
\begin{equation*}
  \mathfrak{n}_0=\mathfrak{m}\;,
\end{equation*}
and the transport equation reads
\begin{equation*}
  \gamma_{\mathfrak{n}_t}(z)\mathrm{d}\mu_{\mathfrak{n}_t}(z) - \gamma_{\mathfrak{m}}(z)\mathrm{d}\mu_{\mathfrak{m}}(z)= i\int_0^t  [\gamma_{\mathfrak{n}_{\tau}}(z),\mathfrak{s}(\tau)]  \alpha_{\tau}(z) \mathrm{d}\mu_{\mathfrak{n}_{\tau}}(z)\mathrm{d}\tau\;.
\end{equation*}
The family of state-valued measures $\{\mathfrak{n}_t\}_{t\in \mathbb{R}}$
given by
\begin{equation*}
  \mathrm{d}\mathfrak{n}_t(z)= \tilde{\mathfrak{U}}_{t,0}(z)\gamma_{\mathfrak{m}}(z)\tilde{\mathfrak{U}}_{t,0}^{*}(z)\;,
\end{equation*}
with $\tilde{\mathfrak{U}}_{t,0}(z)$ the two-parameter unitary group on
$\mathscr{H}$ generated by
\begin{equation*}
  \alpha_{\tau}(z)\mathfrak{s}(\tau)
\end{equation*}
is easily checked to be a solution to the transport equation. Such solution
is actually \emph{unique}, as is proved in a general fashion in
\citep[][Proposition 5.3]{correggi2019arxiv}. Therefore, we have proved that
given
\begin{equation*}
  \Gamma_{\varepsilon_n}\underset{n\to \infty}{\longrightarrow} \mathfrak{m} \quad \mbox{and}\quad  \tr\Bigl(\Gamma_{\varepsilon}\bigl(\mathrm{d}\mathcal{G}_{\varepsilon}(1)+1\bigr)^{1/2}\Bigr)\leq C  \;,
\end{equation*}
there exists a subsequence $\varepsilon_{n_k}$ along which for any $t\in
\mathbb{R}$ we have the convergence
\begin{equation*}
  \Upsilon_{\varepsilon_{n_k}}(t)\underset{k\to \infty}{\longrightarrow} \mathfrak{n}_{t}\;,
\end{equation*}
with
\begin{equation*}
  \mathrm{d}\mathfrak{n}_t(z)= \tilde{\mathfrak{U}}_{t,0}(z)\gamma_{\mathfrak{m}}(z)\tilde{\mathfrak{U}}_{t,0}^{*}(z)\;.
\end{equation*}
By \cref{cor:1}, it also follows that for any $t\in \mathbb{R}$,
\begin{equation*}
  \Gamma_{\varepsilon_{n_k}}(t)\underset{k\to \infty}{\longrightarrow} \mathfrak{m}_{t}\;,
\end{equation*}
with
\begin{equation*}
  \mathrm{d}\mathfrak{m}_t(z)=\mathfrak{U}_{t,0}(z) \gamma(z) \mathfrak{U}^{*}_{t,0}(z) \: \diff \lf( e^{-it\nu\omega}\, _{\star} \,\mu \ri)(z) \;,
\end{equation*}
as stated in \cref{thm:1}. However, a couple of steps are still missing to
complete the proof of the latter.

First of all, one shall prove convergence along the original sequence of
convergence at initial time $\varepsilon_n\to 0$, rather than on some
existing subsequence $\varepsilon_{n_k}\to 0$. This is readily established
exploiting once more the uniqueness of the solution to the transport
equation. Suppose in fact that we have another subsequence
$\varepsilon_{n_j}\to 0$ of convergence for $\Upsilon_{\varepsilon_{n_j}}(t)$
at all times $t\in \mathbb{R}$, with possibly different limit measure
$\{\mathfrak{n}_t'\}_{t\in \mathbb{R}}$. Then, by the same argument as in
\cref{sec:limit-duham-form}, $\mathfrak{n}'_t$ would satisfy the very same
transport equation given in \cref{prop:4} for $\mathfrak{n}_t$. Since the
solution to that transport equation is unique, this would imply
$\mathfrak{n}'_t=\mathfrak{n}_t$. In other words, there is a unique possible
cluster point for the sequence $\Upsilon_{\varepsilon_{n}}(t)$, thus it
converges itself to the very same limit $\mathfrak{n}_t$. We can thus
conclude that, that if
\begin{equation*}
  \Gamma_{\varepsilon_n}\underset{n\to \infty}{\longrightarrow} \mathfrak{m}\quad \mbox{and}\quad  \tr\Bigl(\Gamma_{\varepsilon}\bigl(\mathrm{d}\mathcal{G}_{\varepsilon}(1)+1\bigr)^{1/2}\Bigr)\leq C  \;,
\end{equation*}
then for any $t\in \mathbb{R}$,
\begin{equation*}
  \Gamma_{\varepsilon_{n}}(t)\underset{k\to \infty}{\longrightarrow} \mathfrak{m}_{t}\;,
\end{equation*}
with
\begin{equation*}
  \mathrm{d}\mathfrak{m}_t(z)=\mathfrak{U}_{t,0}(z) \gamma(z) \mathfrak{U}^{*}_{t,0}(z) \: \diff \lf( e^{-it\nu\omega}\, _{\star} \,\mu \ri)(z) \;.
\end{equation*}

\subsection{Relaxing the regularity assumption on the expectation of the
  number operator}
\label{sec:relax-regul-assumpt}

The final step for the proof is to relax the initial time assumption
\begin{equation*}
  \tr\Bigl(\Gamma_{\varepsilon}\bigl(\mathrm{d}\mathcal{G}_{\varepsilon}(1)+1\bigr)^{1/2}\Bigr)\leq C
\end{equation*}
used in the above, to
\begin{equation*}
  \tr\Bigl(\Gamma_{\varepsilon}\bigl(\mathrm{d}\mathcal{G}_{\varepsilon}(1)+1\bigr)^{\delta}\Bigr)\leq C
\end{equation*}
for some $\delta>0$. This is done using standard approximation techniques and
density arguments, as detailed in \citep[][\textsection2]{ammari2011jmpa}. This
concludes the proof of \cref{thm:1}.

{\small \bibliographystyle{natalpha} \bibliography{bib.bib} }

\end{document}